\documentclass[11pt,reqno]{amsart}

\usepackage{amscd,amssymb,amsmath,amsthm}
\usepackage{tikz}
\usepackage{graphicx}
\usepackage{color}
\usepackage{cite}
\newtheorem{thm}{Theorem}
\newtheorem{defn}{Definition}
\newtheorem{lemma}{Lemma}
\newtheorem{pro}{Proposition}
\newtheorem{rk}{Remark}

\numberwithin{equation}{section} 

\newcommand{\bea}{\begin{eqnarray}}
\newcommand{\eea}{\end{eqnarray}}

\def\b{\beta}

\begin{document}
\title[On a new class of Gibbs measures of the Ising model]{Ising model on Cayley trees: a new class of Gibbs measures and their comparison with known ones}

\author{M.M. Rakhmatullaev,  U. A. Rozikov}

\address{M. \ M. \ Rakhmatullaev and U.\ A.\ Rozikov\\ Institute of mathematics,
29, Do'rmon Yo'li str., 100125, Tashkent, Uzbekistan.}
\email {mrahmatullaev@rambler.ru\ \   rozikovu@yandex.ru}

\begin{abstract}
For the Ising  model on Cayley trees we give a very wide class of new Gibbs measures. We show that these new measures are extreme
under some conditions on the temperature. We give a review of all known Gibbs measures of the Ising model on trees
and compare them with our new measures.
 \end{abstract}
\maketitle



\section{Introduction}

The well known nearest neighbors (n.n.) Ising model on the Cayley tree still offers new interesting phenomenon (see e.g.
 \cite{R}, \cite{GHRR} and \cite{GRS} for recent results). Here we widely extend the set of known Gibbs measures of this model.

The Cayley tree $\Gamma^k$
of order $ k\geq 1 $ is an infinite tree, i.e., a connected graph without
cycles, such that exactly $k+1$ edges originate from each vertex.
Let $\Gamma^k=(V, L)$ where $V$ is the set of vertices and  $L$ the set of edges.
Two vertices $x$ and $y$ are called {\it nearest neighbors} if there exists an
edge $l \in L$ connecting them.
We will use the notation $l=\langle x,y\rangle$.
A collection of distinct nearest neighbor pairs $\langle x,x_1\rangle, \langle x_1,x_2\rangle,...,\langle x_{d-1},y\rangle$ is called a {\it path} from $x$ to $y$. The distance $d(x,y)$ on the Cayley tree is the number of edges of the shortest path from $x$ to $y$.

For a fixed $x^0\in V$, called the root, we set
\begin{equation*}
W_n=\{x\in V\,| \, d(x,x^0)=n\}, \qquad V_n=\bigcup_{m=0}^n W_m
\end{equation*}
and denote
$$
S(x)=\{y\in W_{n+1} :  d(x,y)=1 \}, \ \ x\in W_n, $$ the set  of {\it direct successors} of $x$.

 The n.n. Ising model  is then  defined by the
 formal Hamiltonian
\begin{equation}
\label{h}
H(\sigma)=-J\sum_{\langle x,y\rangle\subset V}\sigma(x)\sigma(y).
\end{equation}
Here the first sum runs over   n.n. vertices
$\langle x,y\rangle$, the spins $\sigma(x)$ take values $\pm 1$, and the real parameter  $J$ stands for the interaction energy.

 The (finite-dimensional) Gibbs distributions over configurations    at
 inverse temperature   $\beta=1/T$
 are defined by
 \begin{equation}\label{*}
\mu_n(\sigma_n)=Z^{-1}_n(h)
\exp\Big\{\beta J\sum_{\langle x,y\rangle\subset V_n} \sigma(x)\sigma(y)
+\sum_{x\in W_n}h_x\sigma(x)\Big\}
\end{equation}
  with partition functions given by
 \begin{equation}\label{pf}
Z_n(h)
=
\sum_{\sigma_n }
\exp\Big\{\beta J
\sum_{\langle x,y\rangle\subset V_n} \sigma(x)\sigma(y)
+\sum_{x\in W_n}h_x\sigma(x)\Big\}.
\end{equation}
Here the spin configurations $ \sigma_n $ belong to $\{-1,+1\}^{V_n}$
and
\begin{equation}
h=\{h_x\in \mathbb R, \ \ x\in V\}
\end{equation}
is a collection of real numbers that stands for (generalized) boundary condition.

The probability distributions (\ref{*}) are said compatible if for all
$\sigma_{n-1}$
\begin{equation}\label{**}
\sum_{\omega_n}\mu_n(\sigma_{n-1}, \omega_n)=\mu_{n-1}(\sigma_{n-1})
\end{equation}
where the configurations
$\omega_n$ belong to $\{-1,+1\}^{W_n}$.

It is well known (see Chapter 2 of \cite{R} for a detailed proof)
 that this compatibility condition is satisfied if and only if for any $x\in V$
the following equation holds
\begin{equation}\label{***}
 h_x=\sum_{y\in S(x)}f_{\theta}(h_y),
\end{equation}
where
\begin{equation}
\label{****}
\theta=\tanh(\b J), \quad f_{\theta}(h)
={\rm arctanh}(\theta\tanh h).
\end{equation}

Namely,   for any boundary condition
satisfying the functional equation (\ref{***}) there exists a unique Gibbs measure, the correspondence being one-to-one.

A boundary condition satisfying (\ref{***}) is called \emph{compatible}.

The paper is organized as follows.  The results are given in Section~2.
Section 3 contains a review of all known Gibbs measures of the Ising model on Cayley trees and their comparison with the new measures of this paper.
Proofs are given in Section~4.

\section{Results}

Here we consider the half tree.
Namely the root $x^0$
has $k$ nearest neighbors.

We construct below new solutions of the functional equation (\ref{***}).
Consider the following matrix
$$M=\left(\begin{array}{cccc}
a_1& a_2& a_3& a_4\\
a_2& a_1& a_4& a_3\\
b_1& b_2& b_3& b_4\\
b_2& b_1& b_4& b_3\\
\end{array}\right),$$
where $a_i, b_j$ are non-negative integers and
\begin{equation}\label{ab}
a_1+a_2+a_3+a_4=k, \ \ \ b_1+b_2+b_3+b_4=k.
\end{equation}
This matrix defines the numbers of values $\pm h,$ and $\pm l$ in the set $S(x)$ for
each $h_x\in \{\pm h, \pm l\}$. More precisely, the boundary condition
$h=\{h_x, x\in V\}$ with fields taking values $\pm h, \pm l$ defined by the following steps:

\begin{itemize}
\item[(i)]
if at vertex $x$ we have $h_x=h$, then the function has values
\begin{equation*}
\left\{\begin{array}{cl}
h &  {\rm on}    \quad  a_1  \quad  {\rm vertices  \ of }
\quad S(x);
\\[3mm]
-h & {\rm on} \ \ a_2 \ \  {\rm of \ remaining \ vertices};
\\[3mm]
l & {\rm on}  \ \ a_3 \ \ {\rm of \ remaining \ vertices};\\[3mm]
-l & {\rm on}  \ \ a_4 \ \ {\rm of \ remaining \ vertices}.
\end{array}
\right.
\end{equation*}
\item[(ii)]
if at vertex $x$ we have $h_x=l$, then the function has values
\begin{equation*}
\left\{\begin{array}{cl}
h &  {\rm on}    \quad  b_1  \quad  {\rm vertices  \ of }
\quad S(x);
\\[3mm]
-h & {\rm on} \ \ b_2 \ \  {\rm of \ remaining \ vertices};
\\[3mm]
l & {\rm on}  \ \ b_3 \ \ {\rm of \ remaining \ vertices};\\[3mm]
-l & {\rm on}  \ \ b_4 \ \ {\rm of \ remaining \ vertices}.
\end{array}
\right.
\end{equation*}
\end{itemize}
If at vertex $x$ we have $h_x=-h$ (resp. $-l$) then we multiply the above formulas to $-1$.
(See Fig.\ref{f1} for an example of such function.)
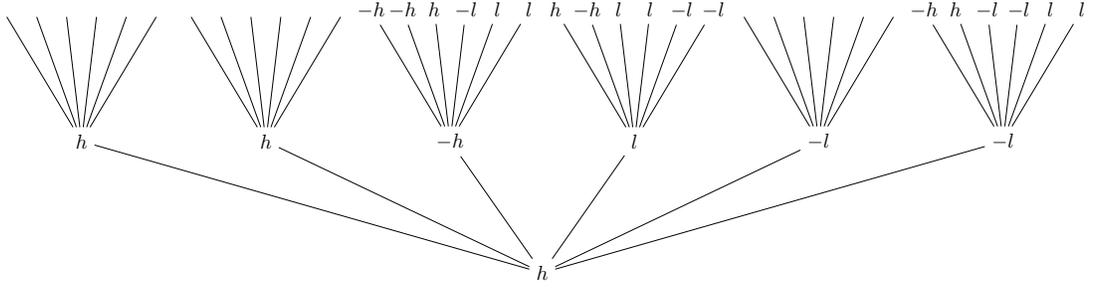
\begin{figure}
\scalebox{.7}{
\begin{tikzpicture}[level distance=2.5cm,
level 1/.style={sibling distance=3.5cm},
level 2/.style={sibling distance=.6cm}]
\node {$h$} [grow'=up]
	child {node {$h$}
		child foreach \name in {0,0,0,0,0,0} { node {$ $} }
	}
	child {node {$h$}
		child foreach \name in {0,0,0,0,0,0} { node {$ $} }
	}
	child {node {$-h$} 
		child foreach \name in {-h,-h,h,-l,l,l} { node {$\name$} }
	}
	child {node {$l$}
		child foreach \name in {h,-h,l,l,-l,-l} { node {$\name$} }
	}
	child {node {$-l$}
		child foreach \name in {0,0,0,0,0,0} { node {$ $} }
	}
	child {node {$-l$}
		child foreach \name in {-h,h,-l,-l,l,l} { node {$\name$} }
	}
;
\end{tikzpicture}
}
\caption{\footnotesize \noindent
In this figure the values of function $h_x$ on the vertices of the Cayley tree of order 6 are shown.
This is the case when $a_1=a_4=2$, $a_2=a_3=1$, $b_1=b_2=1$, $b_3=b_4=2$. The picture shows all four possible rules to put values
of function $h_\cdot$ on the set $S(x)$ when one knows the value of $h_x$ at $x$.}\label{f1}
\end{figure}

It is easy to see that the boundary conditions in the above construction are compatible iff $h$ and $l$ satisfy the following system of equations:
\begin{equation}\label{h12}
\left\{\begin{array}{ll}
h=(a_1-a_2)f_{\theta}(h)+(a_3-a_4)f_\theta(l)\\[3mm]
l=(b_1-b_2)f_{\theta}(h)+(b_3-b_4)f_\theta(l),
\end{array}
\right.
\end{equation}
where $a_i$ and $b_i$ are given in matrix $M$.

Denote
\begin{equation}\label{abcd}
a=a_1-a_2, \ \ b=a_3-a_4, \ \ c=b_1-b_2, \ \ d=b_3-b_4.
\end{equation}
By condition (\ref{ab}) we have $a,b,c,d\in \{-k,-k+1,\dots,k-1,k\}$.
Then the system (\ref{h12}) has the form
\begin{equation}\label{hl}
\left\{\begin{array}{ll}
h=af_{\theta}(h)+bf_\theta(l)\\[3mm]
l=cf_{\theta}(h)+df_\theta(l).
\end{array}
\right.
\end{equation}
\begin{thm}\label{t1}
Independently of the parameters the system of equations (\ref{hl}) has solution $(0,0)$, and if $|(bc-ad)\theta^2+(a+d)\theta|>1$ then there are at least three distinct solutions $(0,0)$,  $(\pm h_*, \pm l_*)$, where $h_*>0, l_*>0$.
\end{thm}

As it was mentioned above,  for any boundary condition
satisfying the functional equation (\ref{***}) there exists a unique Gibbs measure, thus
by solutions $(h,l)$ mentioned in Theorem \ref{t1}, we can construct new Gibbs measures, denoted by $\mu_{h, l}$.
These measures also depend on the choice of the value of the root, and
 differ in cases of non-uniqueness of Theorem \ref{t1}.

\begin{thm}\label{t2} Let $\theta>0$ (i.e., $J>0$, the ferromagnetic Ising model) then
\begin{itemize}
\item[1.] If $hl=0$ then corresponding measure $\mu_{h,l}$ is extreme for $\theta\in ({1\over k}, {1\over \sqrt{k}})$, where the measure $\mu_{h,l}$ exists;

\item[2.] The measures $\mu_{h,l}$, with $h>0$, $l>0$ are extreme as soon as they exist.
\end{itemize}
 \end{thm}

Proofs are given in Section 4.

\section{Relation of the measures $\mu_{h,l}$ to known ones}

 {\it Translation invariant measures.}  (see e.g. \cite{BRZ}, \cite{Ge}, \cite{Pr}) Such measures correspond to $h_x\equiv h$, i.e. constant functions.   These measures are particular cases of our measures mentioned in Theorem \ref{t1} which can be obtained for $a=a_1-a_2=k$, i.e. $a_1=k$, $a_2=a_3=a_4=0$. In this case
 the condition (\ref{***}) reads
\begin{equation}\label{f}
h=kf_\theta(h).
\end{equation}
 The equation (\ref{f}) has a unique solution $h=0$, if $\theta\leq \theta_{\rm c}={1/ k}$ and three distinct solutions $h=0,\pm h_*$ ($h_*>0$), when $\theta >\theta_{\rm c}$.

Let us denote by $\mu_0$, $\mu_{\pm}$ the corresponding Gibbs measures and recall the following known results for the ferromagnetic Ising model ($\theta \geq 0$):
 \begin{itemize}
    \item[(1)] If $ \theta\leq \theta_{c}$,  $\mu_0$ is unique and extreme.

    \item[(2)] If $\theta > \theta_{c}$,   $\mu_-$ and $  \mu_+$, are extreme.

    \item[(3)]  $\mu_0$ is
extreme if and only if $ \theta < {1/\sqrt{k}}$.
    \end{itemize}

{\it ART construction.}

Let $h$ be a boundary condition satisfying (\ref{***}) on
$ \Gamma^{k_0}$.
 For $k\geq k_0+1$ define the following boundary condition on
 $ \Gamma^{k}$:
\begin{equation}\label{1}
\tilde{h}_x=\left\{\begin{array}{ll}
h_x, \ \ \mbox{if} \ \ x\in V^{k_0}\\[2mm]
0, \ \ \mbox{if} \ \ x\in V^k\setminus V^{k_0},\\
\end{array}\right.
\end{equation}
where $V^k$ denote the set of  vertices of $\Gamma^k$.
Namely,
 to each vertices of $V^{k_0}$ one adds $k-k_0$ successors with
vanishing value of the boundary condition.
It is obvious the b.c. $\tilde{h}$ satisfy the compatibility condition (\ref{***}).
In this way one constructs a new set of Gibbs measures that are
extreme in the range  $1/k_0 < \theta < 1/ \sqrt{k}$ (see \cite{Ak} for details).

In case $h$ is translation invariant on $\Gamma^{k_0}$ then the corresponding measures
of this construction can be obtained by Theorem \ref{t1} for $a_1=k_0$, $a_2=b_1=b_2=0$, $a_3+a_4=k-k_0$ and $l=0$.
(See Fig.\ref{f2} for an example.)

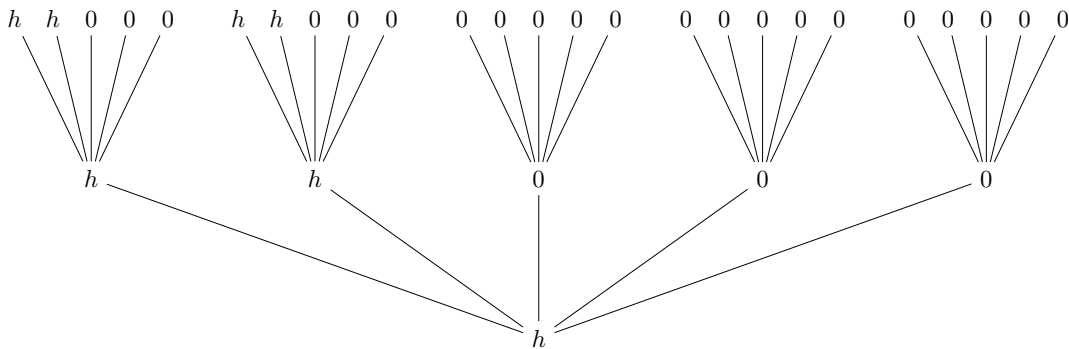
\begin{figure}
\scalebox{.85}{
\begin{tikzpicture}[level distance=2.5cm,
level 1/.style={sibling distance=3.5cm},
level 2/.style={sibling distance=.6cm}]
\node {$h$} [grow'=up]
	child {node {$h$}
		child foreach \name in {h,h,0,0,0} { node {$\name$} }
	}
	child {node {$h$}
		child foreach \name in {h,h,0,0,0} { node {$\name$} }
	}
	child {node {$0$} 
		child foreach \name in {0,0,0,0,0} { node {$\name$} }
	}
	child {node {$0$}
		child foreach \name in {0,0,0,0,0} { node {$\name$} }
	}
	child {node {$0$}
		child foreach \name in {0,0,0,0,0} { node {$\name$} }
	}
;
\end{tikzpicture}
}
\caption{\footnotesize \noindent
This is an example of the function $h_x$ on the vertices of the Cayley tree of order 5.
This is the case when $a_1=2$ $a_3+a_4=3$, $a_2=b_1=b_2=0$, $l=0$. The picture shows all two possible rules to put values
of function $h_\cdot$ on the set $S(x)$, i.e., on n.n. of $h$ one puts two $h$ and three zeros, but on n.n of $0$ one puts only zeros.}\label{f2}
\end{figure}
But in case when $h$ is not translation invariant, measures of ART do not coincide with measures of Theorem \ref{t1}.

{\it Bleher-Ganikhodjaev construction.}
Consider an infinite path $\pi=\{x^0=x_0<x_1<\dots\}$ on the half Cayley tree  (the notation $x<y$ meaning that  pathes from the root to $y$ go through $x$).
Associate  to this  path
a collection $h^\pi$ of numbers
 given  by the condition
\begin{equation}\label{b3.1}
h_x^\pi=\left\{\begin{array}{ll}
-h_*, \ \ \mbox{if} \ \ x\prec x_n, \, x\in W_n,\\[2mm]
h_*, \ \ \mbox{if} \ \ x_n\prec x, \, x\in W_n,\\[2mm]
h_{x_n}, \ \ \mbox{if} \ \ x=x_n.
\end{array}\right.
\end{equation}
$n=1,2,\dots$ where $x\prec x_n$ (resp. $x_n\prec x$) means that $x$ is on the left (resp. right) from the path $\pi$
and $h_{x_n}\in [-h_*, h_*]$ are arbitrary numbers.
For any infinite path $\pi$, the collection of numbers $h^\pi$ satisfying relations (\ref{***})
exists and is unique (see \cite{BG}).

A real number
$t=t(\pi)$, $0\leq t\leq 1$ can be assigned to the infinite path and the set $h^{\pi(t)}$
is uniquely defined.
The  set of numbers $h^{\pi(t)}$ being distinct for different $t\in [0,1]$, it is also the case for the
 corresponding Gibbs  measures. One thus obtains uncountable many  Gibbs measures and they are extreme.
For each fixed $t$ the ground state configuration of such measure contains a unique interface path $\pi(t)$.
Using Theorem \ref{t1} we can construct new class of measures which has infinitely many interface paths.
Let us give these measures precisely:

Let $k\geq 2, a_1\geq 2$  such that $k-a_1$ an even positive integer.
In Theorem \ref{t1} take $l=h$ and $a_2+a_4=a_3$, $b_i=a_i$, $i=1,2,3,4$.
Then corresponding $h_x$ has two values $\pm h$ such that if $h_x=h$ (resp. $-h$) then on $S(x)$ the
number of $h$ (resp. $-h$)
is $a_1+a_3$ and number of $-h$ (resp. $h$) is $a_2+a_4$.
In this case the non-uniqueness condition of Theorem \ref{t1} reduced to
$|\theta|>{1\over a_1}$ under this condition there are two distinct measures
$\mu_{h,-h}$ and $\mu_{-h,h}$. It is easy to see that each such measure has
infinitely many interface pathes which have started point in each level of the tree.
(See Fig. \ref{f3})
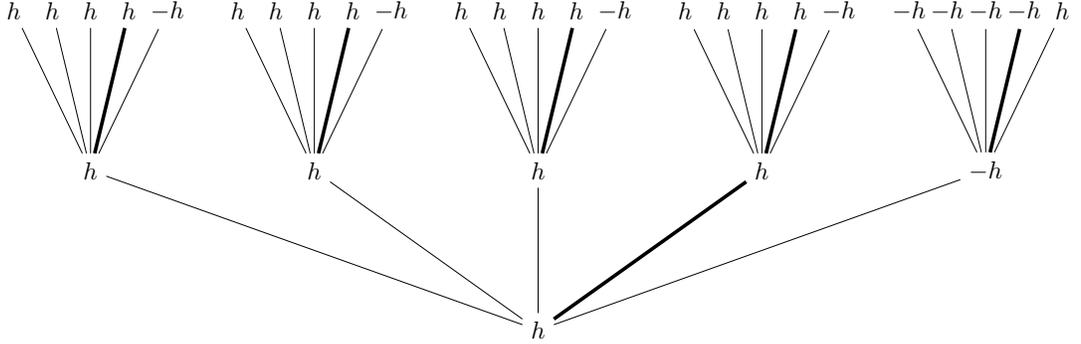
\begin{figure}
\scalebox{.85}{
\begin{tikzpicture}[level distance=2.5cm,
level 1/.style={sibling distance=3.5cm},
level 2/.style={sibling distance=.6cm}]
\node {$h$} [grow'=up]
	child {node {$h$}
		child foreach \name in {h,h,h} { node {$\name$} }
		child {node {$h$} edge from parent[ultra thick]}
		child {node {$-h$}}
	}
	child {node {$h$}
		child foreach \name in {h,h,h} { node {$\name$} }
		child {node {$h$} edge from parent[ultra thick]}
		child {node {$-h$}}
	}
	child {node {$h$} 
		child foreach \name in {h,h,h} { node {$\name$} }
		child {node {$h$} edge from parent[ultra thick]}
		child {node {$-h$}}
	}
	child {node {$h$} edge from parent[ultra thick]
		child foreach \name in {h,h,h} { node {$\name$} edge from
		parent[thin]}
		child {node {$h$} edge from parent[ultra thick]}
		child {node {$-h$} edge from parent[thin]}
	}
	child {node {$-h$}
		child foreach \name in {-h,-h,-h} { node {$\name$} }
		child {node {$-h$} edge from parent[ultra thick]}
		child {node {$h$}}
	}
;
\end{tikzpicture}
}
\caption{\footnotesize \noindent
This is an example of the function $h_x$ on the vertices of the Cayley tree of order 5. Here
$a_1=4$, $a_2=1$, $a_3=a_4=0$. One can take the same function for $a_1=3$, $a_2=0$,  $a_3=a_4=1$, $b_i=a_i$,
$i=1,2,3,4$ and $l=h$. The bold pathes are going to infinity and each is an interphase path (separating "+" and "-" values).}\label{f3}
\end{figure}

We note that such measures (i.e. with infinitely many interfaces of corresponding ground states)
were constructed in \cite{GRS} too. But the methods of \cite{GRS} are different from solving of an equation
with respect to $h_x$. It is known that (see Theorem 12.6 of \cite{Ge}) to each extreme Gibbs measure
corresponds a solution $h_x$ of (\ref{***}). Our Theorem \ref{t1} gives explicitly the solutions corresponding
 to the measures with an infinite pathes of interface. Such solutions corresponding to the extreme measures of \cite{GRS} are not known yet.

{\it Periodic Gibbs measures.}
Let $G_k$ be a free product of $k + 1$ cyclic groups of the second order with generators $a_1, a_2,\dots, a_{k+1}$,
respectively.

It is known that there exists an one-to-one correspondence between the set of vertices $V$ of the
Cayley tree $\Gamma^k$ and the group $G_k$.

\begin{defn} Let ${\tilde G}$ be a normal subgroup of the group $G_k$. The set $h = \{h_x: x\in G_k\}$
 is said to be ${\tilde G}$-periodic if $h_{yx} =h_x$ for any $x\in G_k$ and $y\in {\tilde G}$.
 \end{defn}

Let

$$G^{(2)}_k = \{x\in G_k: \, \mbox{the length of word} \, x \, \mbox{is even}\}.$$
 Note that $G^{(2)}_k$ is  the  set of even vertices (i.e. with even distance to the root). Consider the boundary conditions
 ${h}^{\pm}$ and ${h}^{\mp}$:

\begin{equation}\label{PBC}
{h}_x^{\pm}=- {h}_x^{\mp} =
\left\{\begin{array}{ll}
h_*, \ \ \mbox{if} \ \ x\in G^{(2)}_k\\[2mm]
-h_*, \ \ \mbox{if} \ \ x\in G_k\setminus G^{(2)}_k,\\
\end{array}\right.
\end{equation}
and denote by $\mu^{(\mp)}, \mu^{(\pm)}$ the corresponding Gibbs measures.

The $\tilde{G}$- periodic solutions of equation (\ref{***}) are either translation-invariant ($G_k$-periodic) or $G^{(2)}_k$-periodic (see \cite{GR}), they are solutions to

\begin{equation}\label{ff}
u=kf_\theta(v), \ \ v=kf_\theta(u).
\end{equation}

In the ferromagnetic case only translation invariant b.c. can be found.
In the antiferromagnetic  case ($\theta \leq 0$) the system (\ref{ff}) has
 a unique solution $h=0$ if $ \theta\geq -1/k$, and three distinct solutions $h= 0$, ${h}^{\pm}$ and ${h}^{\mp}$ if
 $ \theta < -1/k$.

    Let us also recall that for the antiferromagnetic Ising model:
    \begin{itemize}
    \item[(1)] If $ \theta  \geq -1/k$,  $\mu_0$ is unique and extreme.

    \item[(2)] If $\theta <- 1/k$,   $\mu^{(\pm)}$ and $  \mu^{(\mp)}$, are extreme.
    \end{itemize}
 see \cite{Ge}, \cite{R}.

We note that these measures are particular cases of measures of Theorem \ref{t1} which can be obtained
for $a_1=0, a_2=k$, i.e $a_3=a_4=0$. (See Fig. \ref{f4}, for $k=3$).
\begin{figure}
\scalebox{1}{
\begin{tikzpicture}[level distance=2.5cm,
level 1/.style={sibling distance=4.5cm},
level 2/.style={sibling distance=1.7cm},
level 3/.style={sibling distance=.6cm}]
\node {$-h$} [grow'=up]
	child {node {$h$}
		child {node {$-h$}
			child {node {$h$}}
			child {node {$h$}}
			child {node {$h$}}
		}
		child {node {$-h$}
			child {node {$h$}}
			child {node {$h$}}
			child {node {$h$}}
		}
		child {node {$-h$}
			child {node {$h$}}
			child {node {$h$}}
			child {node {$h$}}
		}
	}
	child {node {$h$}
		child[sibling distance=.7cm] {node {$-h$}
			child[sibling distance=.6cm] foreach \name in {h,h,h} { node
		{$\name$} }
		}
		child[sibling distance=.7cm] {node {$-h$}}
		child[sibling distance=.7cm] {node {$-h$}}
	}
	child[sibling distance=2.5cm] {node {$h$}
		child[sibling distance=.6cm] foreach \name in {-h,-h,-h} { node
		{$\name$} }
	}
;
\end{tikzpicture}
}
\caption{\footnotesize \noindent
This is an example of two-periodic function $h_x$ on the vertices of the Cayley tree of order 3. Here
$a_1=0$, $a_2=k=3$, $a_3=a_4=0$.}\label{f4}
\end{figure}
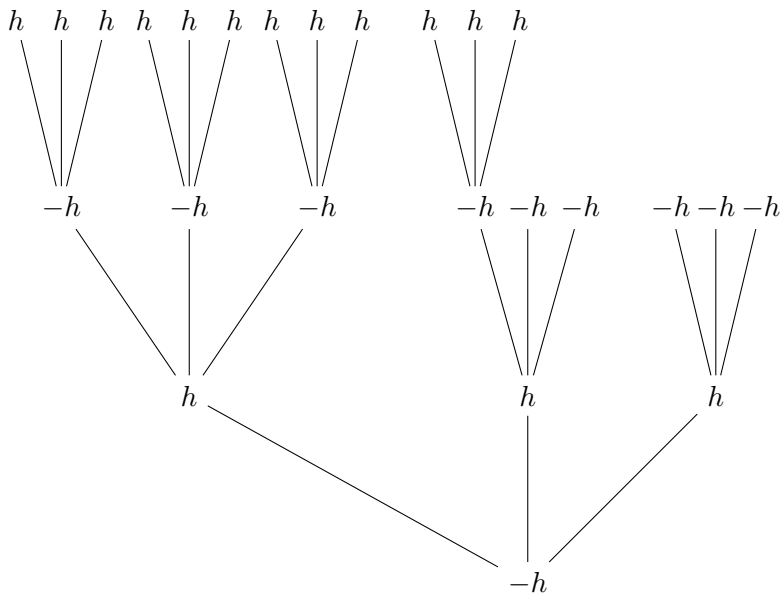

{\it Weakly periodic Gibbs measures.} Following \cite{RR}, \cite{RR1} recall notion of weakly
periodic Gibbs measures.

Let $G_k/\widehat{G}_k=\{H_1,...,H_r\}$  be a factor group, where
$\widehat{G}_k$ is a normal subgroup of index $r\geq 1$.

\begin{defn}\label{wp} A set  $h=\{h_x,x\in G_k\}$ is called
 $\widehat{G}_k$ - \textit{weakly periodic}, if
$h_x=h_{ij}$, for any $x\in H_i, x_{\downarrow}\in H_j$, where $x_{\downarrow}$ denotes the ancestor of $x$.
\end{defn}

Weakly periodic b.c.  $h$ coincide with periodic ones  if
 $h_x$ is independent of $x_{\downarrow}$.

We recall results known for the cases of index two. Note that any such subgroup has the form
\begin{equation}
\label{HA}
H_A=\left\{x\in G_k:\sum\limits_{i\in A}\omega_x(a_i) \ \ {\rm is \, even} \right\},
\end{equation}
where $\emptyset \neq A\subseteq N_k=\{1,2,\dots,k+1\}$, and $\omega_x(a_i)$ is the
number of $a_i$ in a word $x\in G_k$.
We consider  $A\ne N_k$: when  $A = N_k$ weak periodicity coincides
with standard periodicity.

 Let
$G_k/H_A=\{H_0, H_1\}$ be the  factor group, where  $H_0=H_A,
H_1=G_k\setminus H_A$.
 Then, in view of (\ref{***}), the
$H_A$-weakly periodic b.c.   has the form
\begin{equation}\label{wp5}
h_x=\left\{%
\begin{array}{ll}
    h_{1}, & {x \in H_0, \ x_{\downarrow} \in H_0}, \\[2mm]
    h_{2}, & {x \in H_0, \ x_{\downarrow} \in H_1}, \\[2mm]
    h_{3}, & {x \in H_1, \ x_{\downarrow} \in H_0}, \\[2mm]
    h_{4}, & { x \in H_1, \ x_{\downarrow}  \in H_1,}
\end{array}%
\right.\end{equation}
where  the $h_{i}$ satisfy the following equations:
\begin{equation}\label{wp6}
\left\{%
\begin{array}{ll}
    h_{1}=|A|f_\theta(h_{3})+(k-|A|)f_\theta(h_{1}),\\[2mm]
    h_{2}=(|A|-1)f_\theta(h_{3})+(k+1-|A|)f_\theta(h_{1}),\\[2mm]
    h_{3}=(|A|-1)f_\theta(h_{2})+(k+1-|A|)f_\theta(h_{4}),\\[2mm]
    h_{4}=|A|f_\theta(h_{2})+(k-|A|)f_\theta(h_{4}).
\end{array}%
\right.\end{equation}
It is obvious that the following sets are invariant with respect to the operator $W:\mathbb R^4\to \mathbb R^4$ defined by RHS of (\ref{wp6}):
$$ I_1 =\{h\in \mathbb R^4: h_1=h_2=h_3=h_4\}, \ \ I_2 =\{h\in \mathbb R^4:
h_1=h_4; h_2=h_3\},$$
$$
I_3 =\{h\in \mathbb R^4: h_1=-h_4; h_2=-h_3\}.
$$
It is obvious to see that
\begin{itemize}
\item[-] measures corresponding to solutions on $I_1$ are translation invariant, i.e particular cases of the measures given in Theorem \ref{t1}.

\item[-] measures corresponding to solutions on $I_2$ are weakly periodic, which coincide with the measures
given in Theorem \ref{t1} for $a_1=k-|A|$, $a_2=0$, $a_3=|A|$, $a_4=0$, $b_1=k+1-|A|$, $b_2=0$, $b_3=|A|-1$, $b_4=0$.

\item[-] measures corresponding to solutions on $I_3$ are weakly periodic, which coincide with the measures
given in Theorem \ref{t1} for $a_1=k-|A|$, $a_2=0$, $a_3=0$, $a_4=|A|$, $b_1=k+1-|A|$, $b_2=0$, $b_3=0$, $b_4=|A|-1$.
(see Fig. \ref{f5})
\end{itemize}

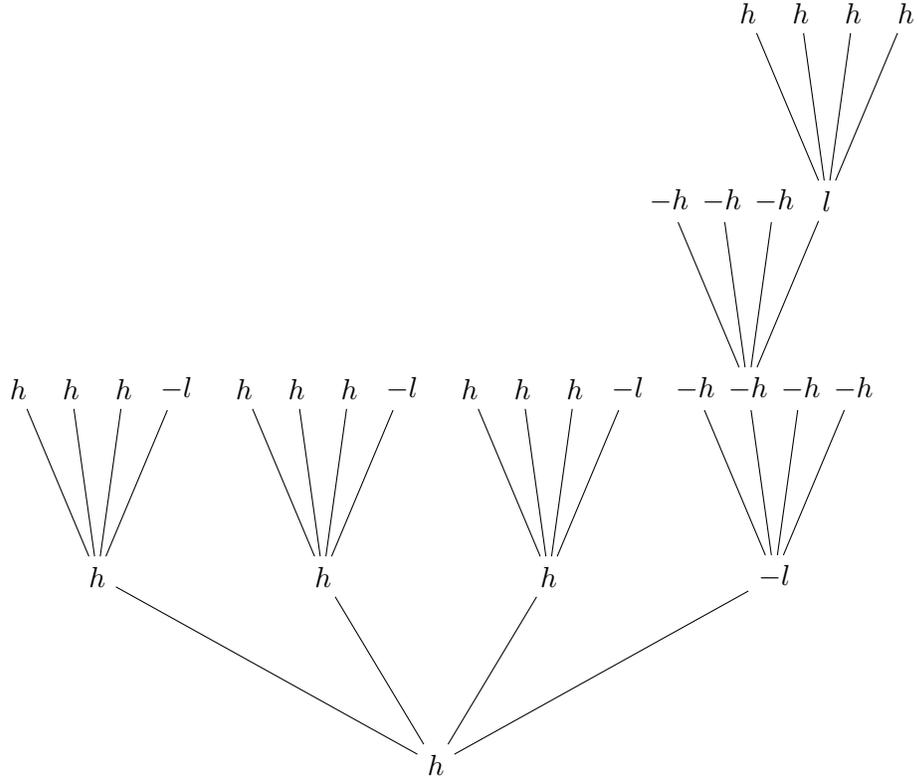
\begin{figure}
\scalebox{1}{
\begin{tikzpicture}[level distance=2.5cm,
level 1/.style={sibling distance=3cm},
level 2/.style={sibling distance=.7cm},
level 3/.style={sibling distance=.7cm},
level 4/.style={sibling distance=.7cm}]
\node {$h$} [grow'=up]
	child {node {$h$}
		child foreach \name in {h,h,h,-l} { node
		{$\name$} }
	}
	child {node {$h$}
		child foreach \name in {h,h,h,-l} { node
		{$\name$} }
	}
	child {node {$h$}
		child foreach \name in {h,h,h,-l} { node
		{$\name$} }
	}
	child {node {$-l$}
		child {node {$-h$}}
		child {node {$-h$}
			child {node {$-h$}}
			child {node {$-h$}}
			child {node {$-h$}}
			child {node {$l$}
				child foreach \name in {h,h,h,h} { node {$\name$} }
			}
		}
		child {node {$-h$}}
		child {node {$-h$}}
	}
;
\end{tikzpicture}
}
\caption{\footnotesize \noindent
This is an example of weakly-periodic function $h_x$ on the vertices of the Cayley tree of order 4. Here
$|A|=1$, $a_1=3$, $a_2=a_3=0$, $a_4=1$, $b_1=k=4$, $b_2=b_3=b_4=0$.}\label{f5}
\end{figure}

Moreover, the system (\ref{wp6}) was solved only in cases $|A|=1$ and $|A|=k$ (see \cite{RR}, \cite{RR1}). Thus Theorem \ref{t1} gives, in particular, new weakly
periodic measures.

\begin{rk} It remain known Gibbs measures called: Zachary measures
(see e.g. part (b) of Theorem 12.31 in \cite{Ge}),
Higuchi's non-translation-invariant
measures (see \cite{Hi}),  Alternating Gibbs measures
(see \cite{GHRR}) and weakly periodic measures for subgroups of index 4 (see Chapter 2 of \cite{R} for details). All these measures correspond to functions $h_x$ with
more than 4 distinct values. Thus these measures are different from the measures mentioned in Theorem \ref{t1}.
\end{rk}

\section{Proofs}

\subsection{Proof of Theorem \ref{t1}}
\begin{proof}
To prove theorem we shall use the following (simply checked)  properties of
the function $f_\theta(x)$.
\begin{lemma}\label{l1} The function $f_\theta$ has the following properties:
\begin{itemize}
\item[1.] $f_\theta(-x)=-f_\theta(x)$, i.e., it is odd function of $x$;\\

\item[2.] $f_{-\theta}(x)=-f_\theta(x)$, i.e., it is odd function of $\theta$;\\

\item[3.] $\lim_{x\to\infty}f_\theta(x)={\rm arctanh}(\theta)$;\\

\item[4.] ${d\over dx}f_\theta(0)=\theta, \ \ 0<{d\over dx}f_\theta(x)\leq\theta$, $\theta>0$;\\

\item[5.] ${d^2\over dx^2}f_\theta(x)<0, \ \ x>0, \ \ \theta>0.$\\

\item[6.] The equation $h=mf_\theta(h)$ (where $m\geq 1$, $\theta\in (-1,1)$) has unique solution $h=0$,
if $-1<\theta\leq {1\over m}$,
and three solutions $h=0,\pm h_*, \, h_*>0$, if ${1\over m}<\theta<1$.
 \end{itemize}
 \end{lemma}

If $a=b=c=d=0$ then the system (\ref{hl}) has unique solution $(h,l)=(0,0)$.
Thus to have non-zero solution it is necessary to have a non-zero parameter, i.e. $a^2+b^2+c^2+d^2>0$.
Since $h$ and $l$ play symmetric role (up to rename of parameters), it suffices to consider the following cases:

1) $a=b=0$. In this case from system (\ref{hl}) we get
$h=0$ and $l=df_\theta(l)$. Part 2 of Lemma \ref{l1} allows us
 to assume $d\geq 0$ only (otherwise we can change $\theta$ by $-\theta$).
The case $d=0$ gives $l=0$. So it remains $d\in \{1,\dots,k\}$. By part 6 of Lemma \ref{l1} we
get (for $a=b=0$) the system (\ref{hl}) has

\begin{itemize}
\item unique solution $(0,0)$,
if $d=0$ or $-1<\theta\leq {1\over d}$;

\item three solutions $(0,0), (0,\pm l_*), \, l_*>0$, if ${1\over d}<\theta<1$, $d\ne 0$.
\end{itemize}
2) $a=0$, $b\ne 0$. In this case from the first equation of the system (\ref{hl}) we get
$h=bf_\theta(l)$. Then from the second equation we obtain
\begin{equation}\label{l}
l=g(l)\equiv cf_\theta(bf_\theta(l))+df_\theta(l).
\end{equation}
Using Lemma \ref{l1} one can see that
$g(0)=0$, $g(-l)=-g(l)$, $g'(0)=bc\theta^2+d\theta$ and $g(l)$ is a bounded function of $l$.
 Moreover, if $|g'(0)|>1$ (i.e. $0$ is unstable fixed point of $g$) then there is a sufficiently small neighborhood of $l=0$: $(-\varepsilon, +\varepsilon)$ such that $g(l)<l$, for $l\in (-\varepsilon, 0)$ and $g(l)>l$, for $l\in (0, +\varepsilon)$.
For $l\in (0,\varepsilon)$ the iterates $g^{(n)}(l)$ remain $>0$, monotonically increase and hence converge to a limit, $l_*\geq 0$ which solves
(\ref{l}). However, $l_*>0$ as $0$ is unstable. Then since $g$ is odd function of $l$,  $-l_*$ also solves
(\ref{l}). Thus

\begin{itemize}
\item If $|g'(0)|=|bc\theta^2+d\theta|>1$ then the system (\ref{hl})
has at least three solutions:
$$(0,0), \ \ (\pm bf_\theta(l_*),\ \pm l_*).$$
\end{itemize}

3) $a\ne 0$, $b=0$. In this case from the first equation of (\ref{hl})
we obtain $h=af_\theta(h)$. As above without loss of generality here
we assume that $a>0$. Then part 6 of Lemma \ref{l1} gives
that the last equation has up to three solutions: 0, $\pm h_*$.
The case $h=0$ reduces the second equation of (\ref{hl}) to
$l=df_\theta(l)$ which is also the equation of the form mentioned in the part 6 of Lemma \ref{l1}.
The cases $h=\pm h_*$ reduces the second equation to
\begin{equation}\label{l5}
 l=\pm cf_\theta(h_*)+df_\theta(l)=\pm {c\over a}h_*+df_\theta(l).
 \end{equation}
Analysis of solutions to this equation done in Lemma 12.27 of \cite{Ge}:
denote
$$\bar h=\max_{l\geq 0}[df_\theta(l)-l].$$
For $\theta>0$ and $d\geq 1$ the equation (\ref{l5}) has
\begin{itemize}
\item[(i)] a unique solution $l_*$ when $|{c\over a}h_*|>\bar h$ or $|{c\over a}h_*|=\bar h=0$,

\item[(ii)] two distinct solutions $l_-<l_+$ when $|{c\over a}h_*|=\bar h>0$, and

\item[(iii)] three distinct solutions $l_-<l_0<l_+$ when $|{c\over a}h_*|<\bar h$.
\end{itemize}
Summarizing we get following nine solutions:
\begin{itemize}
\item If $\theta>{1\over d}$ then there are three solutions to (\ref{hl}):
$$(0,0), \ \ (0,l_*), \ \ (0,-l_*)$$
For $h_*$ satisfying above-mentioned condition (iii) we have
six new solutions
$$(\pm h_*,\pm l_-), \ \ (\pm h_*,\pm l_0),  \ \ (\pm h_*,\pm l_+).$$
\end{itemize}

4) $ab\ne 0$. In this case from the first equation of (\ref{hl}) we get
$$f_\theta(l)={1\over b}(h-af_\theta(h)).$$
Using this from the second equation we obtain
$$l=cf_{\theta}(h)+{d\over b}(h-af_\theta(h))=\varphi(h)\equiv {1\over b}[(bc-ad)f_\theta(h)+dh].$$
Consequently, the first equation of (\ref{hl}) can be written as
\begin{equation}\label{H}
h=\psi(h)\equiv af_{\theta}(h)+bf_\theta\left({1\over b}[(bc-ad)f_\theta(h)+dh]\right).
\end{equation}
It is easy to see that $\psi(0)=0$ and similarly as case 2) one can show that
the equation (\ref{H}) has at least three solutions if
$|\psi'(0)|=|\theta||(bc-ad)\theta+a+d|>1$.
Thus the are at least three solutions to (\ref{hl}) of the form:
$$(h_i,\varphi(h_i)), \ \ i=1,2,3.$$\end{proof}

\subsection{Proof of Theorem \ref{t2}}
\begin{proof}
We use a result of \cite{MSW} to establish a bound for reconstruction insolvability corresponding to the
Gibbs measure $\mu_{h,l}$. Because, it is known that if $\mu$ is  a Gibbs measure of an associated spin
system, the fact that reconstruction is impossible for $\mu$ is equivalent to saying that $\mu$ is an
extremal Gibbs measure of the spin system.

Let us first give some necessary definitions from \cite{MSW}. For $k\geq 2$, let $\mathbb T^k$ denote a {\it half tree}, i.e., the infinite
rooted $k$-ary tree (in which every vertex has $k$ children).  Consider an
{\it initial finite complete subtree} $\mathcal T$, that  is a tree of the
following form: in the rooted tree $\mathbb T^k$, take all vertices at distance $\leq d$ from
the root, plus the edges joining them, where $d$ is a fixed constant. We identify subgraphs of $\mathcal T$ with their vertex sets and write $E(A)$ for the edges within a subset $A$ and $\partial A$ for the boundary of $A$, i.e., the neighbors of $A$ in $(\mathcal T\cup \partial\mathcal T)\setminus A$.

In \cite{MSW} the key ingredients are two quantities,
$\kappa$ and $\gamma$, which bound the probabilities of
percolation of disagreement down and up the tree, respectively. Both are properties of
the collection of Gibbs measures $\{\mu^\tau_{{\mathcal T}}\}$, where the boundary condition $\tau$
is fixed and $\mathcal T$ ranges over all initial finite complete subtrees of $\mathbb T^k$.
 For a given subtree $\mathcal T$ of $\mathbb T^k$ and a vertex $x\in\mathcal T$, we write $\mathcal T_x$ for
the (maximal) subtree of $\mathcal T$ rooted at $x$ that is a tree given by
$\mathcal T\cap \mathbb T^k_x$, with $\mathbb T_x^k$ the half tree with root $x$.
We draw the trees with the root at the
top and the leaves at the bottom.
When $x$ is not the root of $\mathcal T$, let $\mu_{\mathcal T_x}^s$
denote the (finite-volume) Gibbs measure in which the
parent of $x$ has its spin fixed to $s$ and the configuration on the bottom boundary  of ${\mathcal T}_x$
(i.e., on $\partial {\mathcal T}_x\setminus \{\mbox{parent\ \ of}\ \ x\}$) is
specified by $\tau$.

 For two measures $\mu_1$ and $\mu_2$ on $\Omega$, $\|\mu_1-\mu_2\|_x$ denotes the variation distance between the projections of $\mu_1$ and $\mu_2$ onto the spin at $x$, i.e.,
$$\|\mu_1-\mu_2\|_x={1\over 2}(|\mu_1(\sigma(x)=-1)-\mu_2(\sigma(x)=-1)|+|\mu_1(\sigma(x)=1)-\mu_2(\sigma(x)=1)|).$$

Denote by $\Omega^\tau_\mathcal T$
the set of configurations $\sigma$ given on $\mathcal T\cup \partial\mathcal T$ that agree with $\tau$ on $\partial \mathcal T$,
i.e., $\tau$ specifies
a boundary condition on $\mathcal T$. For any $\eta\in \Omega^\tau_\mathcal T$ and any subset
$A\subseteq \mathcal T$, the Gibbs distribution on $A$ conditional on the configuration outside $A$ being $\eta$ is
denoted by $\mu^\eta_A$.

Let $\eta^{x,s}$ be the
configuration $\eta$ with the spin at $x$ set to $s$.

 Following \cite[page 165] {MSW} define
\begin{equation}\label{ka}
\kappa\equiv \kappa(\mu)=\sup_{x\in\Gamma^k}\max_{s,s'}\|\mu^s_{{\mathcal T}_x}-\mu^{s'}_{{\mathcal T}_x}\|_x;
\end{equation}
$$\gamma\equiv\gamma(\mu)=\sup_{A\subset \mathbb T^k}\max\|\mu^{\eta^{y,s}}_A-\mu^{\eta^{y,s'}}_A\|_x,$$
where the supremum is taken over all subsets $A\subset \mathbb T^k$, the maximum is taken over all boundary conditions $\eta$, all sites $y\in \partial A$, all neighbors $x\in A$ of $y$, and all spins $s, s'\in \{-1, 1\}$.

As the main ingredient we apply \cite[Theorem 4.3]{MSW}, from which it follows that
the Gibbs measure $\mu$ is extreme if $k\kappa\gamma<1$.

To use the above-mentioned condition for the given choices of solutions to (\ref{***}) we have
to bound corresponding $\kappa$ and $\gamma$ and show that $k\kappa\gamma <1$.

For both $\kappa$ and $\gamma$, we need to bound a quantity of the form $\|\mu^\eta_A -\mu^{\eta^y}_A\|_z$, where
$y\in \partial A$ and $z\in A$ is a neighbor of $y$. The key observation of \cite{MSW} is that this quantity can be
expressed very cleanly in terms of the "magnetization" at $z$, i.e., the ratio of probabilities
of a $(-)$-spin and a $(+)$-spin at $z$. It will actually be convenient to work with the magnetization
without the influence of the neighbor $y$: let $\mu_A^{\eta,*}$ denote the Gibbs
distribution with boundary condition $\eta$, except that the spin at $y$ is free (or equivalently,
the edge connecting $z$ to $y$ is erased).

\begin{pro}\label{p1} Let $\mu$ be one of measures $\mu_{h,l}$.
For any subset $A\subseteq \mathcal T$, any boundary configuration $\eta$, any site $y\in \partial A$
and any neighbor $z\in A$ of $y$, we have
$$\|\mu^\eta_A -\mu^{\eta^y}_A\|_z=K_\beta(R_z),$$
where
$$R_z = {\mu_A^{\eta,*}(\sigma(z)=-)\over \mu_A^{\eta,*}(\sigma(z)=+)}=e^{-2h_z},$$
here $h_z$ is a compatible function constructed in Section 2 using $h$ and $l$ by steps (i)-(ii).
The function $K_\beta$ is defined by
$$K_\beta(a) ={1\over e^{-2\beta J} a+1}-{1\over e^{2\beta J}a+1}.$$
\end{pro}
\begin{proof}  The prove is similar to the proof of \cite[Proposition 4.2]{MSW}.
\end{proof}

Note that $R_z\in [0,+\infty)$. It is easy to check that $K_\beta(a)$ is an increasing function in the interval $[0, 1]$,
decreasing in the interval $[1, +\infty]$, and is maximized at $a = 1$. Therefore, we can always
bound $\kappa$ and $\gamma$ from above by $K_\beta(1) =\theta=\tanh(\beta J)$.
Thus the bounds of $\kappa$ and $\gamma$ can be controlled by the magnetization $R_z$.
The bound will be better than $\theta$ when $R_z$ differs
from 1 for any $z$.

To prove part 1 of Theorem \ref{t2} we use estimates $\gamma\leq \theta$ and $\kappa\leq \theta$ because $hl=0$
gives that $R_z=1$ for some $z\in V$. Thus condition $k\gamma \kappa<1$ gives $k\theta^2<1$ and the part 1 follows.

Now we shall prove the part 2.
For the Gibbs measure $\mu_{h,l}$ corresponding to a solution $(h,l)$ of (\ref{hl}) we denote
$$\mathcal H^{\pm}=\{x\in V: h_z=\pm h\}.$$
$$\mathcal L^{\pm}=\{x\in V: h_z=\pm l\}.$$
$$\alpha=e^{-\beta J}, \ \ A=e^{2h}, \ \ C=e^{2l}.$$
$$F(x)={\alpha+x\over 1+\alpha x}.$$
Then $R_z$ corresponding to $\mu_{h,l}$ has the following form
$$R_z=\left\{\begin{array}{lll}
A, \ \ z\in \mathcal H^+\\[2mm]
1/A, \ \ z\in \mathcal H^-\\[2mm]
C, \ \ z\in \mathcal L^+\\[2mm]
1/C, \ \ z\in \mathcal L^-,
\end{array}\right.$$
where $A\ne 1$, $C\ne 1$ (since $h\ne 0$ and $l\ne 0$) and satisfy the following system of equations
\begin{equation}\label{bu}
\begin{array}{lll}
A=[F(A)]^{a_1-a_2}[F(C)]^{a_3-a_4}, \\[2mm]
C=[F(A)]^{b_1-b_2}[F(C)]^{b_3-b_4}.
\end{array}
\end{equation}

 To check the extremality condition $k\kappa\gamma<1$ for $\mu_{h,l}$
 we use estimation $\gamma<\theta$. To bound $\kappa$ we use that
$R_z$ has values $A, 1/A, C, 1/C$.
Thus we have
$$\kappa\leq \max\{K_\beta(s): s\in\{A, 1/A, C, 1/C\}\}.$$
Without loss of generality we take  $K_\beta(A)=\max\{K_\beta(s): s\in\{A, 1/A, C, 1/C\}\}$,
because $A$ and $C$ play similar role.
We shall use the following formula (see Lemma 4.3 of \cite{MSW}):
$$K_\beta(A)={1\over k}\cdot {A\over J(A)}\cdot J'(A),$$
where $J(x)=(F(x))^k$. Note that under condition $|\theta||(bc-ad)\theta+a+d|>1$
the solution $h={1\over 2}\ln A$ of (\ref{H}) is an attracting (stable) fixed point for $\psi$.
Moreover, it is known that $h\leq h^*$ (see proof of part 1) of Theorem 12.31 in \cite{Ge}),
where $h^*>0$ is solution to $h=kf_\theta(h)$ (for $\theta>{1\over k}$). Note that $e^{2h^*}$ is
an attractive fixed point of $J(x)$, i.e. $J'(e^{2h^*})<1$. Since $0<A=e^{2h}\leq e^{2h^*}$
we have $J(A)\geq A$ and $J'(A)\leq 1$ for $\theta>{1\over k}$.
Consequently, we get
$$\kappa\leq K_\beta(A)={1\over k}\cdot {A\over J(A)}\cdot J'(A)\leq {1\over k}.$$
Hence
$k\gamma \kappa\leq \theta<1.$
\end{proof}

\section*{ Acknowledgements}

U.A. Rozikov thanks Professor Y.Velenik and the Section de Math\'ematiques
Universit\'e de Gen\'eve for financial support and kind hospitality.

\end{document}